\documentclass[10pt, twocolumn]{IEEEtran}
\usepackage{graphicx}
\usepackage{cite}
\usepackage{multirow}
\usepackage{epstopdf}

\usepackage{amsthm}
\usepackage{amsmath}
\usepackage{amsfonts}
\usepackage{amssymb}
\usepackage{algorithm}
\usepackage{algorithmic}
\usepackage{verbatim}
\usepackage{multirow}
\usepackage{framed}
\usepackage{balance}

\usepackage{graphicx}
\usepackage{epstopdf}
\graphicspath{{fig/}}
\usepackage{calc}
\usepackage{color}
\usepackage[caption=false,font=footnotesize]{subfig}

\newtheorem{theorem}{Theorem}
\newtheorem{lemma}{Lemma}
\newtheorem{remark}{Remark}

\newtheorem{example}{Example}
\newtheorem{definition}{Definition}
\newtheorem{proposition}{Proposition}
\newtheorem{stopping criterion}{stopping criterion}

\newtheorem{problem}{Problem}

\newcommand{\beq}{\begin{equation}}
\newcommand{\eeq}{\end{equation}}
\newcommand{\beqnn}{\begin{equation*}}
\newcommand{\eeqnn}{\end{equation*}}
\newcommand{\beqy}{\begin{eqnarray}}
\newcommand{\eeqy}{\end{eqnarray}}
\newcommand{\beqynn}{\begin{eqnarray*}}
\newcommand{\eeqynn}{\end{eqnarray*}}
\newcommand{\bit}{\begin{itemize}}
\newcommand{\eit}{\end{itemize}}
\newcommand{\ben}{\begin{enumerate}}
\newcommand{\een}{\end{enumerate}}
\newcommand{\bex}{\begin{example}}
\newcommand{\eex}{\end{example}}

\newcommand{\balg}[1]{\begin{algorithm} \caption{#1}}
\newcommand{\ealg}{\end{algorithm}}

\newcommand{\balgc}{\begin{algorithmic}[1]}
\newcommand{\ealgc}{\end{algorithmic}}

\newcommand{\bary}{\begin{array}}
\newcommand{\eary}{\end{array}}
\newcommand{\bmx}{\begin{bmatrix}}
\newcommand{\emx}{\end{bmatrix}}
\newcommand{\bsmx}{\left[\begin{smallmatrix}}
\newcommand{\esmx}{\end{smallmatrix}\right]}
\newcommand{\bmxc}[1]{\left[\begin{array}{@{}#1@{}}}
\newcommand{\emxc}{\end{array}\right]}
\newcommand{\bcn}{\begin{center}}
\newcommand{\ecn}{\end{center}}

\newcommand{\Rbb}{{\mathbb{R}}}

\newcommand{\bigO}{{\mathcal{O}}}

\newcommand{\Rnbn}{\Rbb^{n \times n}}

\newcommand{\A}{\boldsymbol{A}}
\newcommand{\B}{\boldsymbol{B}}

\renewcommand{\H}{\boldsymbol{H}}
\newcommand{\I}{\boldsymbol{I}}

\renewcommand{\L}{\boldsymbol{L}}
\newcommand{\M}{\boldsymbol{M}}

\newcommand{\Q}{\boldsymbol{Q}}
\newcommand{\R}{\boldsymbol{R}}

\newcommand{\U}{\boldsymbol{U}}

\newcommand{\Z}{\boldsymbol{Z}}
\renewcommand{\a}{\boldsymbol{a}}
\renewcommand{\b}{\boldsymbol{b}}

\renewcommand{\d}{\boldsymbol{d}}

\renewcommand{\l}{\boldsymbol{l}}

\newcommand{\rr}{\boldsymbol{r}}
\newcommand{\s}{\boldsymbol{s}}
\renewcommand{\u}{\boldsymbol{u}}
\renewcommand{\v}{\boldsymbol{v}}

\newcommand{\x}{{\boldsymbol{x}}}

\newcommand{\0}{{\boldsymbol{0}}}

\newcommand{\br}{{\bar{r}}}

\newcommand{\bbR}{{\bar{\R}}}

\newcommand{\bxi}{\boldsymbol{\xi}}

\newcommand{\8}{{\infty}}
\newcommand{\4}{{\nabla}}

\newcommand{\norm}[1]{\left\lVert #1 \right\rVert} 

 \usepackage{cite,color}
\DeclareMathOperator*{\argmin}{arg\,min}

\usepackage{slashbox}

\begin{document}

\title{An Efficient Optimal Algorithm for the Successive Minima Problem}

\author{Jinming~Wen, Lanping~Li, Xiaohu Tang, \IEEEmembership{Member, IEEE}, and Wai Ho Mow, \IEEEmembership{Senior Member, IEEE}
\thanks{Part of this work was presented at the  2017 IEEE International Conference on Communications (ICC), Paris, France.}
\thanks{J.~Wen is with  the College of Information Science and Technology, and the College of Cyber Security, Jinan University, Guangzhou 510632, China (e-mail: jinming.wen@mail.mcgill.ca).}
\thanks{L. Li and X. Tang are with the Information Security and National Computing Grid
Laboratory, Southwest Jiaotong University, Chengdu 610031, China (e-mail: lanping@my.swjtu.edu.cn, xhutang@swjtu.edu.cn).}
\thanks{W. H. Mow is with the Department of Electronic and Computer Engineering, Hong Kong University of Science and Technology, Hong Kong (e-mail: eewhmow@ust.hk).}
\thanks{
This work was partially supported by National Natural Science Foundation of China (No. 11871248),
``the Fundamental Research Funds for the Central Universities'' (No. 21618329),
the University Grants Committee of Hong Kong under Project AoE/E-02/08
and the Major Frontier Project of Sichuan Province under Grant 2015JY0282.}
}

\maketitle

\begin{abstract}
In many applications including integer-forcing linear multiple-input and multiple-output (MIMO)
receiver design, one needs to solve a  successive minima problem (SMP) on an $n$-dimensional lattice to
get an optimal integer coefficient matrix $\A^\star\in \mathbb{Z}^{n\times n}$.
In this paper, we first propose an  efficient optimal SMP  algorithm with  an $\bigO(n^2)$ memory complexity.
The main idea behind the new algorithm is it first initializes with a suitable suboptimal solution,
which is then updated, via a novel algorithm with only $\bigO(n^2)$ flops in each updating,
until $\A^{\star}$ is obtained.
Different from existing algorithms which find $\A^\star$ column by column
through using a sphere decoding search strategy $n$ times,
the new algorithm uses a search strategy once only.
We then rigorously prove the optimality of the proposed algorithm.
Furthermore, we theoretically analyze its complexity.
In particular, we not only show that the new algorithm is $\Omega(n)$ times faster than
the most efficient existing algorithm with polynomial memory complexity,
but also assert that it is even more efficient than the most efficient existing algorithm
with exponential memory complexity.
Finally, numerical simulations are presented to illustrate the optimality
and efficiency of our novel SMP algorithm.
\end{abstract}


\IEEEpeerreviewmaketitle
\section{Introduction}
\label{sec:Introduction}

\subsection{Successive Minima Problem}

Let $\H\in \mathbb{R}^{m\times n} (m\geq n)$ be an arbitrary full-column-rank matrix,
then the lattice $\mathcal{L}(\H)$ generated by $\H$ is defined by
$\mathcal{L}(\H)=\{\H\a|\a \in \mathbb{Z}^n\}$.
For $1\leq k\leq n$, the $k$-th successive minimum $\lambda_k(\H)$ of $\mathcal{L}(\H)$
is defined as the smallest $r$ such that the closed $n$-dimensional ball
$\mathbb{B}(\0,r)$ of radius $r$ centered at the origin contains $k$ linearly independent lattice vectors.

In many applications, such as integer-forcing (IF) linear multiple-input and multiple-output (MIMO)
receiver design  \cite{ZhaNEG10} \cite{ZhaNEG14}, we need to find an invertible \footnote{In this paper, an $n\times n$ invertible
matrix means this matrix is invertible over $\Rnbn$.} matrix
$\A^{\star}=[\a_1^{\star},\a_2^{\star},\cdots, \a_n^{\star}]\in \mathbb{Z}^{n\times n}$
such that $\|\H\a_i^{\star}\|_2$ are as small as possible for $1\leq i \leq n$.
By the definition of successive minima, finding such matrix $\A^{\star}$ is equivalent to solving
a successive minima problem (SMP) on lattice $\mathcal{L}(\H)$, which is mathematically defined
in the following:
\begin{definition}
\label{d:SMP}
SMP on lattice $\mathcal{L}(\H)$: finding an invertible matrix $\A^{\star}=[\a^{\star}_1,\ldots, \a^{\star}_{n}]\in \mathbb{Z}^{n\times n}$ such that
\[
\|\H\a^{\star}_i\|_2= \lambda_i(\H),\quad i=1,2,\ldots, n.
\]
\end{definition}

Note that IF linear MIMO receiver reaches the optimal diversity-multiplexing tradeoff for the
standard MIMO channel with no coding across transmit antennas in the high signal-to-noise ratio (SNR)
region \cite{ZhaNEG14}.
In addition to the IF linear receiver, there are many other efficient and effective
MIMO detection approaches, such as the likelihood ascent search detector whose complexity is $\bigO(mn)$ flops \cite{VarMCR08},
minimum mean-squared error iterative successive
parallel arbitrated decision feedback detectors whose exact complexity analysis was not given \cite{LamS08},
unified bit-based probabilistic data association detection approach
whose complexity is substantially lower than that of the conventional symbol-based
probabilistic data association detectors in uncoded V-BLAST systems using high-order QAM \cite{YanLMH11},
energy spreading transform approach whose complexity is $\bigO(mn)$ flops \cite{HwaKP12},
decision-feedback-based algorithm whose complexity is $\bigO(mn^2)$ flops \cite{LiL12},
adaptive and iterative multi-branch minimum mean-squared
error decision feedback detection algorithms whose complexity is $\bigO(mn^2)$ flops \cite{Lam13},
and iterative detection and decoding algorithms for low-density parity-check codes
whose exact complexity analysis was not given \cite{UchHL16}.
Different from the IF linear receiver, which decodes integer combinations of the transmitted codewords
based on the fact that any integer linear combination of lattice points is still a lattice point,
these detectors decode the transmitted codewords individually.
When the channel matrix is near singular, the IF linear receiver may have better detection performance
than these approaches.

In addition to the IF linear receiver design, solving an SMP is needed in some other applications,
such as physical-layer network coding framework design \cite{FenSK13} and compute-compress-and-forward relay strategy design \cite{TanY16}.
Motivated by these applications, this paper focuses on developing an efficient algorithm for
optimally solving the SMP and analyzing its complexity.

\subsection{Related works}

There are several optimal SMP algorithms \cite{WeiC12a,MejO13,FenSK13, DinKWZ15, FisCS16}.
For conciseness, we briefly introduce the main ideas of the algorithms
in the recent two papers \cite{DinKWZ15, FisCS16} only.
There are two SMP algorithms in \cite{DinKWZ15} which are respectively for solving SMP's on real and complex lattices. For conciseness, we introduce the algorithm for the real SMP only.
For efficiency, this algorithm first utilizes the  Lenstra-Lenstra-Lov\'asz (LLL) reduction \cite{LenLL82}
to preprocess the SMP.
Then, as in \cite{FenSK13}, it finds the transformed $\A^\star$  column by column,
by an improved Schnorr-Euchner search algorithm \cite{SchE94} which is a widely used sphere decoding search algorithm, in $n$ iterations.
Finally, this solution is left multiplied  with the unimodular matrix generated by the LLL reduction to give $\A^\star$.
The first and last steps of the algorithm in \cite{FisCS16} are the same as those of the algorithm
in \cite{DinKWZ15}. However, its second step is different.
Specifically, it first creates a matrix $\M$,
which stores a list of sorted (in an nondecreasing order according to their lengths) lattice vectors
with lengths bounded by the largest length of all the column vectors of the
LLL reduced matrix of $\H$.
These vectors are obtained by employing the Alg. ALLCLOSESTPOINTS in \cite{AgrEVZ02}.
Then $\M$ is transformed into the row echelon form by the Gaussian elimination.
Finally, the first $n$ linearly independent columns of $\M$ are chosen to
form the transformed $\A^\star$.
Although simulations in this paper will show that the latter is much more efficient than
the former, different from the former, its memory complexity is
an exponential function in $n$.
Thus, a more efficient algorithm with polynomial memory complexity is still desirable.

There are also some suboptimal SMP algorithms, such as the slowest descent method
\cite{WeiC13} and lattice reduction based algorithms \cite{SakHV13, LyuL17}.

\subsection{Our Contributions}

In this paper, we develop an efficient optimal SMP algorithm for $\A^{\star}$.
Specifically, the contributions of this paper are summarized as follows:
\begin{itemize}
\item An efficient optimal  algorithm for the SMP is proposed
\footnote{This part was presented at the  2017 IEEE International Conference on Communications (ICC)~\cite{WenLTMT17b}, but the  strategy for updating the suboptimal solution
is further improved in this version.}.
Like existing ones, for efficiency, our algorithm first uses the LLL reduction \cite{LenLL82}
to preprocess the SMP by reducing $\H$.
Then, unlike \cite{DinKWZ15}, which forms the solution of the transformed SMP column by column in $n$ iterations,
our new algorithm initializes with a suboptimal solution which is an $n\times n$
permutation matrix such that it is a fairly  good initial solution of the transformed SMP.
The suboptimal  solution is then updated by a novel algorithm which uses  the improved Schnorr-Euchner search algorithm in \cite{WenC17b}
to search for candidates of the columns of the solution of the transformed SMP,
and uses a novel and efficient algorithm to update it.
The updating process continues until the optimal solution is obtained.
Finally, the solution of the transformed SMP is left multiplied  with the unimodular matrix generated by the LLL reduction to give the optimal solution $\A^\star$ (see Section \ref{S:PM}).

\item We theoretically show that the memory complexity and the expected time complexity of
our new algorithm are respectively $\bigO(n^2)$ space and
$\bigO(n^{3/2}(2\pi e\,m/n)^{n/2})$ flops (see Section \ref{SS:compAnaA2}).
\item We show that the new algorithm
is $\Omega(n)$ \footnote{ Here $\Omega(n)$ is the standard big omega notation; see for example https://en.wikipedia.org/wiki/Big-O-notation
for its detailed defintion.}
times faster than the algorithm in \cite{DinKWZ15},
which is the most efficient existing algorithm with polynomial
memory complexity (see Section \ref{SS:compD}).
We also assert that it is faster than the one in \cite{FisCS16}
whose memory complexity is exponential in $n$ (see Section \ref{SS:compF}).

\item Numerical simulations not only verify the
improvements as predicted from the above theoretical findings
but also show that the proposed optimal algorithm  is more efficient than the
Minkowski reduction based suboptimal algorithm for solving the SMP (see Section \ref{S:Sim}).
\end{itemize}

The rest of the paper is organized as follows.
We propose our new SMP optimal algorithm and show its optimality in Section \ref{S:PM}.
We analyze its space and time complexities in Section \ref{S:compAna}.
Simulation results are provided in Section \ref{S:Sim} to show the efficiency and superiority
of the proposed algorithm.
Finally, conclusions are given in Section \ref{S:Conclusions}.

{\it Notation.}
Let $\mathbb{R}^{m\times n}$ and $\mathbb{Z}^{m\times n}$ respectively
stand for the spaces of the $m\times n$ real and integer matrices.
Let $\mathbb{R}^n$ and $\mathbb{Z}^n$ denote the spaces of the $n$-dimensional  real
and integer column vectors, respectively.
Matrices and column vectors are respectively denoted by uppercase and lowercase letters.
For a matrix $\A$, let $a_{ij}$ denote its element at row  $i$ and column  $j$,
$\a_i$ denote its $i$-th column, $\a_{k,i:j}$ be the vector formed by $a_{ki},a_{k,i+1}, \ldots, a_{kj}$
and $\A_{i:j}$ be the submatrix of $\A$ formed by columns from $i$ to $j$.
For a vector $\x$, let $x_i$ be its $i$-th element
and $\x_{i:j}$ be the subvector of $\x$ formed by entries $i, i+1, \ldots,j$.
For a number $x$, let $\lfloor x\rceil$ denote its nearest integer
(if there is a tie, the one with smaller magnitude is chosen).

\section{A New SMP Algorithm}
\label{S:PM}

In this section, we propose an efficient algorithm for exactly solving the SMP
and rigorously show its optimality.

The main novelties of our new algorithm are: one the one hand, unlike \cite{DinKWZ15},
which needs to use an improved Schnorr-Euchner search algorithm  $n$ times,
it uses the improved Schnorr-Euchner search algorithm \cite{WenC17b} once only.
On the other hand, the complexity of the novel algorithm for updating the suboptimal solution is
$\bigO(n^2)$ flops only. As will be explained in details in Section \ref{SS:Preliminary},
in comparison, a straightforward updating algorithm costs $\bigO(n^4)$ flops.
Because of these two novelties, the efficiency of the new algorithm is significantly improved.

\subsection{Preprocessing of the SMP}

For efficiency, one typically uses the LLL reduction \cite{LenLL82} to preprocess the SMP.
Let $\H$  have the following thin QR factorization whose algorithms can be found in many references (see, e.g., \cite{GolV13}):
\beq
\label{e:QR}
\H=\Q\R,
\eeq
where $\Q\in \mathbb{R}^{m\times n}$ is a matrix with orthonormal columns,
$\R \in \mathbb{R}^{n\times n}$ is an upper triangular matrix.
Recall that we assume $\H$ is a full column matrix, so $\R$ is full-rank.

After obtaining $\R$, the LLL reduction  reduces $\R$ in \eqref{e:QR} to $\bbR$ through
\beq
\label{e:QRZ}
\bar{\Q}^T \R \Z = \bbR,
\eeq
where $\bar{\Q}  \in \mathbb{R}^{n \times n}$ is orthogonal,
$\Z\in   \mathbb{Z}^{n \times n}$ is  unimodular (i.e., $\Z$ also satisfies $|\det(\Z)|=1$),
and $\bbR\in \mathbb{R}^{n \times n}$ is an upper triangular matrix satisfying
\begin{align*}
&|\br_{ik}|\leq\frac{1}{2} |\br_{ii}|, \quad i=1, 2, \ldots, k-1,  \\
&\delta\, \br_{k-1,k-1}^2 \leq   \br_{k-1,k}^2+ \br_{kk}^2,\quad k=2, 3, \ldots, n,
\end{align*}
where $\delta$ is a constant satisfying $1/4 < \delta \leq 1$.
The matrix $\bbR$ is said to be  LLL reduced.
The LLL reduction algorithm can be found in \cite{LenLL82} and its properties in MIMO communications have been studied in \cite{ChaWX13,WenTB16,WenC17}.

By \eqref{e:QR} and \eqref{e:QRZ}, we have
$
\H=\Q\bar{\Q}\bar{\R}\Z^{-1} .
$
Since the columns of $\Q$ are orthonormal and $\bar{\Q}$ is orthogonal, we have
\[
\|\H\a^{\star}_i\|_2=\|\bbR\Z^{-1}\a^{\star}_i\|_2, \quad 1\leq i\leq n.
\]
By the definition of successive minima, we also have
\[
\lambda_i(\R)=\lambda_i(\bbR), \quad 1\leq i\leq n.
\]
Thus, by Definition \ref{d:SMP}, the SMP on lattice $\mathcal{L}(\H)$ can be transformed to
the SMP on lattice $\mathcal{L}(\bbR)$, i.e., Problem \ref{P:RSMP} below:
\begin{problem}
[SMP on lattice $\mathcal{L}(\bbR)$]\label{P:RSMP}
finding an invertible integer matrix
$\B^{\star}=[\b^{\star}_1,\ldots, \b^{\star}_n]\in \mathbb{Z}^{n \times n}$ such that
\[
\|\bbR\b^{\star}_i\|_2= \lambda_i(\bbR),\quad 1\leq i\leq n.
\]
\end{problem}

Furthermore, if $\B^{\star}$ is a solution to Problem \ref{P:RSMP}, then $\A^{\star}=\Z\B^{\star}$
is a solution to the SMP on lattice $\mathcal{L}(\H)$, which is defined in Definition \ref{d:SMP}.
Moreover, according to the definition of successive minima,  the solution $\B^{\star}$ of
Problem \ref{P:RSMP} satisfies
\beq
\label{e:SMorder}
\|\bbR\b^{\star}_1\|_2\leq \|\bbR\b^{\star}_2\|_2\leq \cdots\leq \|\bbR\b^{\star}_n\|_2.
\eeq

Note that the reason for transforming the SMP on lattice $\mathcal{L}(\H)$ to the SMP on lattice
$\mathcal{L}(\bbR)$ (i.e., Problem \ref{P:RSMP}) is that the latter can be solved much more efficiently
than the former due to the fact that $\bbR$ is LLL reduced.

\subsection{Updating strategy for the novel algorithm}
\label{SS:Preliminary}

The main idea behind the proposed algorithm for Problem \ref{P:RSMP} is as follows:
we start with a suboptimal solution $\B$ which is the $n\times n$ permutation matrix such that
\beq
\label{e:ordernormR}
\|\bbR\b_{1}\|_2\leq \|\bbR\b_{2}\|_2\leq \ldots \leq\|\bbR\b_{n}\|_2.
\eeq
Then, we update the suboptimal solution.
More specifically, we use the improved Schnorr-Euchner search algorithm in \cite{WenC17b}
to find a nonzero integer vector $\b$ satisfying $\|\bbR\b\|_2<\|\bbR\b_{n}\|_2$.
Then, we use it to update $\B$ and then go to the next updating.
More specifically, we use $\b$ and $\B$ to get another suboptimal solution (i.e., an invertible matrix)
denoting by $\bar{\B}$, whose columns are chosen from $\b$ and the columns of $\B$ such that
$\|\bbR\bar{\b}_{i}\|_2$ are as small as possible for $1\leq i\leq n$
and
\beq
\label{e:ordernormRbar}
\|\bbR\bar{\b}_{1}\|_2\leq \|\bbR\bar{\b}_{2}\|_2\leq \ldots \leq\|\bbR\bar{\b}_{n}\|_2.
\eeq
The updating process continues with the process of the improved Schnorr-Euchner enumeration
until the suboptimal solution cannot be updated anymore, and the final solution is $\B^{\star}$.
Note that to ensure the suboptimal solution can be updated to the optimal solution $\B^{\star}$
which satisfies \eqref{e:SMorder},
$\B$ should satisfy \eqref{e:ordernormR} and $\bar{\B}$ should satisfy \eqref{e:ordernormRbar}
through the whole updating process.

From the above analysis, we can see that the updating process is  equivalent to solving Problem \ref{P:problem} below.
\begin{problem}
\label{P:problem}
For a given full-rank matrix $\B\in \mathbb{Z}^{n\times n}$ and a nonzero vector $\b\in\mathbb{Z}^{n}$
which satisfy
\beq
\label{e:ordernorm}
r_1\leq r_2\leq \ldots \leq r_n\mbox{ and } \alpha<r_n,
\eeq
where we denote
\beq
\label{e:norm}
r_k=\|\bbR\b_{k}\|_2,\,\; 1\leq k\leq n, \mbox{ and }\alpha= \|\bbR\b\|_2;
\eeq
find an invertible matrix $\bar{\B}\in \mathbb{Z}^{n\times n}$
whose column vectors are chosen from $\b$ and columns of $\B$ such that
\beq
\label{e:ordernormBbar}
\bar{r}_1\leq \bar{r}_2\leq \ldots \leq \bar{r}_n,
\eeq
and $\bar{r}_i$ are as small as possible for $1\leq i\leq n$,
where we denote
\beq
\label{e:normbar}
\bar{r}_k=\|\bbR\bar{\b}_{k}\|_2,\,\; 1\leq k\leq n.
\eeq
\end{problem}

In the following, we will propose an algorithm with at most $5n^2-2n-1$ flops
(i.e., the summation of the numbers of addition, subtraction, multiplication and division)
to solve Problem \ref{P:problem}.
In comparison, we will explain in detail that a straightforward method for Problem \ref{P:problem}
costs $\bigO(n^4)$ flops. Before giving the details, we need to show the problem is well-defined, i.e., showing the following proposition.

\begin{proposition}
\label{P:welldefine}
Problem \ref{P:problem} is solvable.
\end{proposition}

To prove Proposition \ref{P:welldefine}, we need to introduce the following theorem:

\begin{theorem}
\label{t:iindep}
Let $\B\in \mathbb{Z}^{n\times n}$ be an arbitrary full-rank integer matrix and $\b\in\mathbb{Z}^{n}$ be an arbitrary nonzero integer vector
such that $\widetilde{\B}_{1:i}$ is full column rank for some  $i$ with $1\leq i\leq n$, where
\beq
\label{e:Btilde}
\widetilde{\B}=
\bmx
\b_1&\ldots&\b_{i-1}&\b&\b_{i}&\ldots&\b_{n}
\emx.
\eeq
Then there exists at least one $j$ with $i+1\leq j\leq n+1$
such that $\widetilde{\B}_{[\setminus j]}$ is also full-rank,
where $\widetilde{\B}_{[\setminus j]}$ is the matrix obtained by removing the $j$-th column of $\widetilde{\B}$.
\end{theorem}

\begin{proof}
Since $\B$ is full-rank, there exist $f_i\in \mathbb{R}, 1\leq i \leq n$, such that
\[
\b=f_1\b_1+\ldots +f_n\b_{n}.
\]
Since $\widetilde{\B}_{1:i}$ is full column rank, there exists at least one
$j$ with $i+1\leq j\leq n+1$ such that $f_{j-1}\neq 0$. Otherwise, by the aforementioned equation,
$\b$ is a linear combination of $\b_1,\ldots, \b_{i-1}$ which implies that
$\widetilde{\B}_{1:i}$ is not full column rank, contradicting the assumption.
Then, we have
\begin{align*}
      &\begin{pmatrix}
       \b_1 & \b_2&\ldots &  \b_{j-2}& \b & \b_{j} & \ldots & \b_{n}
     \end{pmatrix}\\
     =&\quad\B
    \begin{pmatrix}
      1 &        &   & f_{1}  &   &   &   \\
        & \ddots &   & \vdots   &   &   &   \\
        &        & 1 & f_{j-2}  &   &   &   \\
        &        &   & f_{j-1}    &   &   &   \\
        &        &   & f_{j} & 1 &   &   \\
        &        &   &\vdots   &   & \ddots &   \\
        &        &   & f_n     &   &   & 1 \\
    \end{pmatrix}
\end{align*}
which implies that
\[
\begin{pmatrix}
       \b_1 & \b_2&\ldots &  \b_{j-2}& \b  & \b_{j} & \ldots & \b_{n}
     \end{pmatrix}
\]
is full-rank. Thus, $\widetilde{\B}_{[\setminus j]}$ is full-rank.
\end{proof}

\begin{remark}
\label{co:iindep}
Note that if $i=1$, $\widetilde{\B}_{1:i}$ reduces to $\b$ which is full column rank
(since $\b\neq \0$), and Theorem \ref{t:iindep} reduces to:

Let $\B\in \mathbb{Z}^{n\times n}$ be an arbitrary full-rank integer matrix and $\b\in\mathbb{Z}^{n}$
be an arbitrary nonzero integer vector, then there exists at least one $j$ with $2\leq j\leq n+1$
such that  the matrix obtained by removing the $j$-th column of
$
\bmx
\b&\b_1&\ldots&\b_{n}
\emx
$
is also full-rank.
\end{remark}

In the following, we prove Proposition \ref{P:welldefine}.

\begin{proof}
By \eqref{e:ordernorm} and \eqref{e:norm}, one can see that there exists an $i$
with $1\leq i \leq n$ such that $r_{i-1}\leq \alpha<r_{i}$
(note that if $i=1$, it means $\alpha<r_{1}$). Let
\beq
\label{e:rtilde}
\widetilde{\rr}=
\bmx
r_1&\ldots&r_{i-1}&\alpha&r_{i}&\ldots& r_{n}
\emx.
\eeq
Then, one can see that
\beq
\label{e:ordernormBtilde}
\widetilde{r}_{1}\leq \widetilde{r}_{2}\leq \ldots \leq\widetilde{r}_{n+1}.
\eeq

Hence, solving Problem \ref{P:problem} is equivalent to finding the largest $j$ with $1\leq j\leq n+1$ such that $\widetilde{\B}_{[\setminus j]}$ is full-rank, where $\widetilde{\B}$ is defined in \eqref{e:Btilde}.
Specifically, after finding $j$, set $\bar{\B}=\widetilde{\B}_{[\setminus j]}$
and $\bar{\rr}=\widetilde{\rr}_{[\setminus j]}$, where $\widetilde{\rr}_{[\setminus j]}$
is the vector obtained by removing the $j$-th entry from $\widetilde{\rr}$.
Then, by \eqref{e:norm}, \eqref{e:Btilde} and \eqref{e:rtilde}, one can see that \eqref{e:normbar} holds.
Moreover, by \eqref{e:ordernormBtilde}, one can see that $\bar{r}_k$ are as small as possible
for $1\leq k \leq n$, and  by \eqref{e:rtilde} and \eqref{e:ordernormBtilde}, \eqref{e:ordernormBbar} holds.

If $\widetilde{\B}_{1:i}$ is not full column rank, then $j=i$, i.e., $\b$ should be removed from $\widetilde{\B}$, and the resulting matrix is $\B$ which is full-rank by assumption.
This is because, no matter which $\widetilde{\b}_j$ is removed for $i+1\leq j\leq n+1$,
the resulting matrix contains $\widetilde{\B}_{1:i}$ as a submatrix,
and hence it is not full-rank.
On the other hand, if $\widetilde{\b}_j$ is removed for $1\leq j\leq i-1$, then it is not
the column with the largest index needs to be removed.

If $\widetilde{\B}_{1:i}$ is full column rank, then by Theorem \ref{t:iindep},
there exists at least one $j$ with $i+1\leq j\leq n+1$ such that $\widetilde{\B}_{[\setminus j]}$
is full-rank.

Thus, there exists a $j$ with $i\leq j\leq n+1$ such that $\widetilde{\B}_{[\setminus j]}$ is full-rank
no matter whether $\widetilde{\B}_{1:i}$ is full column rank or not.
Therefore, Problem \ref{P:problem} is solvable.
\end{proof}

By the above proof, a natural method to find the desired $j$ is to check whether
$\widetilde{\B}_{[\setminus k]}$ is full-rank for $k=n+1, n,\ldots, i+1$
until finding an invertible matrix $\widetilde{\B}_{[\setminus j]}$ if it exists.
Otherwise, i.e., $\widetilde{\B}_{[\setminus k]}$ is not full-rank for $k=n+1, n,\ldots, i+1$,
then by the above analysis, $j=i$.
Clearly, this approach works, but the main drawback of this method is that its worst complexity
is $\bigO(n^4)$ flops which is too high.
Concretely, in the worst case, i.e., if $i=1$ and
$\widetilde{\B}_{[\setminus j]}$ are not full-rank for $2\leq j\leq n+1$,
then $n$ matrices need to be checked.
Since the complexity of checking whether an $n\times n$ matrix is full-rank or not is
$\bigO(n^3)$ flops, the whole complexity is $\bigO(n^4)$ flops.

In the following, we propose a method which solves Problem \ref{P:problem}
in $\bigO(n^2)$ flops, under the assumption that $\B$ has the LU factorization
$\L\B=\U$ with given $\L$ and $\U$, by using an updating LU factorization algorithm.
Specifically, we have the following theorem:
\begin{theorem}
\label{t:alg}
Let $\B$ be defined in Problem \ref{P:problem}, suppose that $\B$ has the following LU factorization:
\beq
\label{e:LU}
\L\B=\U,
\eeq
where full-rank matrix $\L\in \mathbb{R}^{n\times n}$ is a product of lower triangular
and permutation matrices
and $\U\in \mathbb{R}^{n\times n}$ is an invertible upper triangular matrix.
Then there exists an algorithm with  at most $5n^2-2n-1$ flops to find $\bar{\B}$
and $\bar{\rr}$ which satisfy the requirement of Problem \ref{P:problem},
and full-rank matrices $\bar{\L}\in \mathbb{R}^{n\times n}$ and $\bar{\U}\in \mathbb{R}^{n\times n}$,
which are respectively a product of lower triangular and permutation matrices and
upper triangular matrix, such that
\beq
\label{e:LUbar}
\bar{\L}\bar{\B}=\bar{\U}.
\eeq
\end{theorem}

Before proving Theorem \ref{t:alg}, we make a comment.
\begin{remark}
Theorem \ref{t:alg} shows that Problem \ref{P:problem} can be solved by an algorithm with  at most
$5n^2-2n-1$ flops under the condition that the LU factorization of $\B$ (see \eqref{e:LU}) is given
(note that the LU factorization of $\B$ is usually expressed as $\B=\L\U$,
but as can be seen from the following proof, it is better to use $\L\B=\U$).
Recall that the initial $\B$ of our new algorithm for Problem \ref{P:RSMP} is a permutation matrix,
(see the first paragraph of this subsection),
so \eqref{e:LU} holds  by setting $\L=\B$ and $\U=\I$.
Furthermore, Theorem \ref{t:alg} shows that, in addition to returning $\bar{\B}$ and $\bar{\rr}$
as required in Problem \ref{P:problem}, the algorithm also
returns $\bar{\L}$ and $\bar{\U}$ such that \eqref{e:LUbar} holds,
thus assuming that we have the LU factorization of $\B$ is not a problem.
\end{remark}

\begin{proof}
In the following, we propose an algorithm to find $\bar{\B}$, $\bar{\rr}$,
$\bar{\L}$ and $\bar{\U}$ satisfying the requirement of Theorem \ref{t:alg},
and then show its complexity is  at most  $5n^2-2n-1$ flops.

Our method for Problem \ref{P:problem} consists of three steps.
Firstly, we find $i$ with $r_{i-1}\leq \alpha<r_{i}$
such that \eqref{e:rtilde} and \eqref{e:ordernormBtilde} hold
to form $\widetilde{\B}$ (see \eqref{e:Btilde}).
Then, we find the largest $j$  such that
$\widetilde{\B}_{[\setminus j]}$ is full-rank.
Finally, we get $\bar{\B}$ by setting $\bar{\B}=\widetilde{\B}_{[\setminus j]}$,
obtain $\bar{\rr}$ by setting  $\bar{\rr}=\widetilde{\rr}_{[\setminus j]}$
and update $\L$ and $\U$ to $\bar{\L}$ and $\bar{\U}$, respectively.

The first step is trivial, so in the following, we consider the second step.
By \eqref{e:LU}, we have
\beq
\label{e:Utilde}
\L\bmx \B&\b\emx=\bmx \U&\L\b\emx:=\tilde{\U},
\eeq
where we denote $\tilde{\U}=\bmx \U&\L\b\emx\in \mathbb{R}^{n\times(n+1)}$.
Since $\L$ is an invertible matrix and $j\geq i$ (see the proof of Proposition \ref{P:welldefine}),
finding the desired $j$ is equivalent to finding the largest $j$ with $i\leq j\leq n$
such that $\widetilde{\U}_{[\setminus j]}$ is full-rank if it exists.
Specifically, by the proof of Proposition \ref{P:welldefine},
if there exist $i\leq j\leq n$ such that $\widetilde{\U}_{[\setminus j]}$
is full-rank, then the largest $j$ is the one we need; otherwise,
the last column of $\bmx \B&\b\emx$ should be removed to form $\bar{\B}$.
Since $\tilde{\U}_{1:n}$ is an upper triangular matrix, by \eqref{e:Utilde},
to find the desired $j$, we only need to check whether $\tilde{u}_{k,n+1}\neq0$
for $k=n,n-1,\ldots,i$ until finding one $j$ such that $\tilde{u}_{j,n+1}\neq0$ if it exists.
Otherwise, i.e., $\tilde{u}_{k,n+1}=0$ for $k=n,n-1,\ldots,i$,
then we set $\bar{\B}=\B$, $\bar{\rr}=\rr$, $\bar{\L}=\L$ and $\bar{\U}=\U$.

In the following, we introduce the last step.
If $\tilde{u}_{k,n+1}=0$ for $k=n,n-1,\ldots,i$, then by the above analysis,
$\bar{\B}=\B$, $\bar{\rr}=\rr$, $\bar{\L}=\L$ and $\bar{\U}=\U$.
Otherwise,  there exist at least one $i<j\leq n$ such that $\tilde{u}_{j,n+1}\neq0$.
Then we set $\bar{\B}=\widetilde{\B}_{[\setminus j]}$
and $\bar{\rr}=\widetilde{\rr}_{[\setminus j]}$.
To get $\bar{\L}$, $\bar{\U}$, we first update $\tilde{\U}$ to $\hat{\U}$ by setting
\beqnn
\hat{\U}=\bmx \U_{1:i-1}&\L\b& \U_{i:j-1}&\U_{j+1:n}\emx.
\eeqnn
Then, we use elementary transformations of matrices to bring $\hat{\U}$
back to an upper triangular matrix by transforming $\hat{u}_{ki}=0$
for $k=j,\ldots, i+1$ to get $\bar{\U}$.
Meanwhile, we use the same elementary transformations to update
$\L$ to get $\bar{\L}$.

To make readers implement the above algorithm easily, we describe the pseudocode of the above algorithm
in Alg. \ref{a:updatingC}.

\begin{algorithm}
\caption{An efficient updating LU factorization algorithm for updating $\B$}  \label{a:updatingC}
\textbf{Input:} An invertible matrix $\B\in \mathbb{Z}^{n\times n}$, nonzero vectors
$\rr\in\mathbb{R}^{n}$ and $\b\in\mathbb{Z}^{n}$, and positive number $\alpha $ that satisfy
\eqref{e:ordernorm} and \eqref{e:norm},
an invertible upper triangular $\U$ and full-rank matrix $\L$ such that \eqref{e:LU} holds.

\textbf{Output:} An invertible matrix $\bar{\B}$ and a nonzero vector $\bar{\rr}$
satisfy the requirement of Problem \ref{P:problem},
an invertible upper triangular matrix $\bar{\U}$ and an invertible matrix $\bar{\L}$ such that \eqref{e:LUbar} holds.
\begin{algorithmic}[1]
  \STATE set $n=$length$(\b)$;
  \STATE  $\v=\L\b$;\,\;// use $\v$ to store $\L\b$ (see \eqref{e:Utilde})
  \STATE find $i$ such that $r_{i-1}\leq \alpha<r_{i}$;
  \STATE find $j$ which is the largest $k$ such that $v_k\neq0$;
   \IF{$j< i$}
   \STATE set $\bar{\B}=\B$, $\bar{\rr}=\rr$, $\bar{\L}=\L, \bar{\U}=\U$; 

    \ELSE
    \STATE set $\bar{\L}=\L$;
     \STATE set $\bar{\U}=\bmx \U_{1:i-1}&\v& \U_{i:j-1}&\U_{j+1:n}\emx$ // we denote $\v=\L\b$ (see line 2)
     \STATE  set $\bar{\B}=\bmx \B_{1:i-1}&\b& \B_{i:j-1}&\B_{j+1:n}\emx$
     \STATE set $\bar{\rr}=\bmx \rr_{1:i-1}&\alpha& \rr_{i:j-1}&\rr_{j+1:n}\emx$

\IF{$j>i$}
    \FOR{$k=j:-1:i+1$}
      \IF{$\bar{u}_{k-1,i}=0$}
             \STATE swap the $(k-1)$-th and $k$-th rows of $\bar{\U}$, and $(k-1)$-th and $k$-th rows of $\bar{\L}$
           \ELSE
           \STATE set $t=\bar{u}_{ki}/\bar{u}_{k-1,i}$, $\bar{u}_{ki}=0$;
           \STATE $\bar{\u}_{k,k:n}=\bar{\u}_{k,k:n}-t\bar{\u}_{k-1,k:n}$;
           \STATE $\bar{\l}_{k,1:n}=\bar{\l}_{k,1:n}-t\bar{\l}_{k-1,1:n}$;
           \ENDIF
    \ENDFOR
   \ENDIF
   \ENDIF
\end{algorithmic}
\end{algorithm}

In the following, we analyze the complexity of Alg. \ref{a:updatingC} by counting the number of flops.
Since $\L$ may not be a lower triangular matrix, the first 11 steps cost at most $2n^2$ flops (for
computing $\v$), and the last 12 steps cost
\[
2\sum^{j}_{k=i+1}(2n-k+1.5)\leq 2\sum^{n}_{k=2}(2n-k+1.5)=3n^2-2n-1
\]
flops. Thus, the total complexity of Alg. \ref{a:updatingC} is at most $5n^2-2n-1$ flops.
From the above analysis, we can see that, if $i$ is large or $j$ is small,
the total complexity of Alg. \ref{a:updatingC} is much smaller than $5n^2-2n-1$ flops.
\end{proof}

\begin{example}
\label{ex:updatingC}
Let
\beq
\label{e:RCc}
\bbR=\bmx
2&0&0&0\\
0&2&1&0\\
0&0&2&0\\
0&0&0&2
\emx,\,\;
\B=\bmx
1&0&0&0\\
0&2&0&0\\
0&0&3&0\\
0&0&0&4
\emx,\,\;
\b=\bmx
1\\-1\\1\\0
\emx,
\eeq
$\rr=\bmx2&4&3\sqrt{5}&8\emx^T$ and $\alpha=3$.
Then, \eqref{e:ordernorm} and \eqref{e:norm} hold.
In the following, we show how to obtain $\bar{\B}$, $\bar{\rr}$,
$\bar{\L}$ and $\bar{\U}$ such that $\bar{\L}\bar{\B}=\bar{\U}$.

By \eqref{e:RCc}, one can easily obtain that
\[
\L=\bmx
1&0&0&0\\
0&1&0&0\\
0&0&1&0\\
0&0&0&1
\emx,\,\;
\U=\bmx
1&0&0&0\\
0&2&0&0\\
0&0&3&0\\
0&0&0&4
\emx
\]
such that \eqref{e:LU} holds.
Since $\rr=\bmx2&4&3\sqrt{5}&8\emx^T$ and $\alpha=3$, by line 3 of Alg. \ref{a:updatingC}, $i=2$.
By line 2 of Alg. \ref{a:updatingC} $\v=\bmx1&-1&1&0\emx^T$, thus $j=3>i$
(see line 4 of Alg. \ref{a:updatingC})). Then, by lines 9-11 of Alg. \ref{a:updatingC}, we obtain
\[
\bar{\U}=\bmx
1&1&0&0\\
0&-1&2&0\\
0&1&0&0\\
0&0&0&4
\emx,\,\;
\bar{\B}=\bmx
1&1&0&0\\
0&-1&2&0\\
0&1&0&0\\
0&0&0&4
\emx,
\,\;
\bar{\rr}=\bmx
2\\3\\4\\8
\emx.
\]
Since $j=3>2=i$, by lines 13-21 of Alg. \ref{a:updatingC}, we get
\[
\bar{\L}=\bmx
1&0&0&0\\
0&1&0&0\\
0&1&1&0\\
0&0&0&1
\emx,\,\;
\bar{\U}=\bmx
1&1&0&0\\
0&-1&2&0\\
0&0&2&0\\
0&0&0&4
\emx.
\]
It is easy to check that $\bar{\B}$ and $\bar{\rr}$  satisfy \eqref{e:ordernormBbar}
and \eqref{e:normbar}, and $\bar{\L}\bar{\B}=\bar{\U}$.
\end{example}

\subsection{The improved Schnorr-Euchner search algorithm in \cite{WenC17b}}
\label{ss:SE}
From the first paragraph of Sec. \ref{SS:Preliminary}, our novel algorithm for Problem \ref{P:RSMP} needs to use the integer vectors
obtained by the improved Schnorr-Euchner search algorithm \cite{WenC17b}  to update the suboptimal solution $\B$,
thus we briefly review this algorithm in this subsection.
To better explain this algorithm,
we first introduce the Schnorr-Euchner search algorithm \cite{Sch87} for solving the following
shortest vector problem (SVP)
\beq
\label{e:SVPR}
\b^{\star}=\min_{\b\in\mathbb{Z}^{n}\backslash \{0\}}\|\bbR\b\|_2.
\eeq
More details on this algorithm are referred to  \cite{AgrEVZ02, ChaH08},
and its variants can be found in  \cite{CuiHT13, WenZMC16}.

Suppose that $\b^{\star}$ is within the hyper-ellipsoid defined by
\beq
\label{ine:ellp}
\norm{\bar{\R}\b}_2< \beta,
\eeq
where $\beta$ is a given constant. Let
\beq
\label{e:d}
d_{n}=0, \; d_i=-\frac{1}{\bar{r}_{ii}}\sum_{j=i+1}^{n}\bar{r}_{ij}b_j, \quad i=n-1,\ldots, 1.
\eeq
Then \eqref{ine:ellp} can be transformed to
$$
\sum_{i=1}^{n}\bar{r}_{ii}^2(b_i-d_i)^2< \beta^2
$$
which is equivalent to
\begin{eqnarray}
\label{e:levelk}
\bar{r}_{ii}^2(b_i-d_i)^2< \beta^2-\sum_{j=i+1}^{n}\bar{r}_{jj}^2(b_j-d_j)^2
\end{eqnarray}
for $i=n, n-1,\ldots, 1$, where $i$ is called as the level index,
and we define $\sum_{j=n+1}^{n}\bar{r}_{jj}^2(b_j-d_j)^2:=0$.

The Schnorr-Euchner search algorithm starts with $\beta=\infty$,
and sets $b_i=\lfloor d_i\rceil$ ($d_{i}$ are computed via \eqref{e:d}) for $i=n, n-1,\ldots, 1$.
Clearly, $\b=\0$ is obtained and \eqref{e:levelk} holds.
Since $\b^{\star}\neq \0$, $\b$ should be updated.
To be more specific, $b_1$ is set as the next closest integer to $d_1$.
Since $\beta=\infty$, \eqref{e:levelk} with $i=1$ holds.
Thus, this updated $\b$ is stored and $\beta$ is updated to $\beta=\norm{\R\b}_2$.
Then, the algorithm tries to update the latest found $\b$ by finding a new $\b$ satisfying \eqref{ine:ellp}.
Since \eqref{e:levelk} with $i=1$ is an equality for the current $\b$,
$b_1$ only cannot be updated.
Thus we try to update $b_2$ by setting it as the next closest integer to $d_2$.
If it satisfies \eqref{e:levelk} with $i=2$, we try to  update $b_1$
by setting $b_1 = \lfloor d_1 \rceil$
and then check whether \eqref{e:levelk} with $i=1$ holds or not, and so on;
Otherwise, we try to update $b_3$, and so on. Finally, when we are not able to find a new integer
$\b$ such that \eqref{e:levelk} holds with $i=n$, the search process stops and
outputs the latest $\b$, which is actually $\b^{\star}$ satisfying \eqref{e:SVPR}.

The improved Schnorr-Euchner search algorithm \cite{WenC17b} is a simple modification
of the Schnorr-Euchner search algorithm based on the fact that if $\b^{\star}$
is an optimal solution to \eqref{e:SVPR}, then so is $-\b^{\star}$.
Specifically, the algorithm in \cite{WenC17b} only searches the nonzero integer vectors
$\b$, satisfying $b_{n}\geq 0$ and $b_i\geq 0$ if $\b_{i+1:n}=\0$ ($1\leq i\leq n-1$).
Note that only the former property of $\b$  is exploited in \cite{DinKWZ15},
whereas  this strategy can prune more vectors while retaining optimality.

\subsection{A novel optimal SMP algorithm}
\label{SS:New alg}

In this subsection, we develop a novel and efficient algorithm for SMP on lattice $\mathcal{L}(\H)$.
We begin with designing the algorithm for Problem \ref{P:RSMP}
by incorporating Alg. \ref{a:updatingC} into the Schnorr-Euchner search algorithm.

The proposed algorithm for Problem \ref{P:RSMP} is described as follows:
we start with a suboptimal solution $\B$ which is the $n\times n$ permutation matrix
satisfying \eqref{e:ordernormR} and assume $\beta=\|\bbR\b_{n}\|_2$.
We further assume $\L=\B$, $\U$ is the $n\times n$ identity matrix,
$\rr\in \mathbb{R}^{n}$ with $r_{k}=\|\bbR\b_{k}\|_2$ for $1\leq k\leq n$.
Clearly, $\rr$ satisfies \eqref{e:ordernorm} and \eqref{e:norm}, and \eqref{e:LU} holds
(note that $\B$ is a permutation matrix).
Then we use the improved Schnorr-Euchner search algorithm \cite{WenC17b} to search
nonzero integer vectors $\b$ satisfying \eqref{ine:ellp} to update $\B$.
Specifically,
we denote $\alpha=\|\bbR\b\|_2$ (since $\b$ satisfies \eqref{ine:ellp} and $\beta=\|\bbR\b_{n}\|$,
$\alpha$ satisfies \eqref{e:ordernorm}), use Alg. \ref{a:updatingC} to
update $\B$, $\rr$, $\L$ and $\U$ and set $\beta=r_n$.
After this, we go to the next updating,
i.e., we use the improved Schnorr-Euchner search algorithm to update $\b$ and $\alpha$
(more details on how to update $\b$ is referred to Sec. \ref{ss:SE}),
and then use Alg. \ref{a:updatingC} to update $\B$, $\rr$, $\L$ and $\U$.
Finally, when $\B$ cannot be updated anymore and $\beta$ cannot be decreased anymore
(i.e., when we are not able to find a new value for $b_{n}$
such that \eqref{e:levelk} holds with $k=n$), the search process stops and outputs $\B^{\star}$.

By the above analysis, the proposed algorithm for Problem \ref{P:RSMP}
can be summarized in Alg. \ref{a:RSMP}, where
\begin{eqnarray}
\label{e:sgn}
\mbox{sgn}(x)=
\begin{cases}
1, & x\geq0 \cr
-1, &x<0
 \end{cases}.
\end{eqnarray}

\begin{algorithm}
 \caption{A sphere decoding based algorithm for Problem \ref{P:RSMP}}\label{a:RSMP}
 \textbf{Input:} A nonsingular upper triangular  $\bbR \in \mathbb{R}^{n \times n}$.

 \textbf{Output:} A solution $\B^{\star}$ to Problem \ref{P:RSMP}.

\begin{algorithmic}[1]
 \STATE $n=$length$(u)$;
   \STATE $\b=\0$;
   \STATE $\d = \0$ ($n\times 1$ zero vector); // see \eqref{e:d}
    \STATE $\rr=\0$;
  \FOR{i=1:n}
  \STATE    $r_i=\|\bar{\rr}_i\|_2$; //  $r_i=\|\bbR\b_i\|_2$;
  \ENDFOR
  \STATE $[\rr,\mbox{order}]=\mbox{sort}(\rr)$, $\beta=r_n$;  //initial radius, see \eqref{ine:ellp}
 \STATE  $\B=\I_n$, $\B=\B(:,\mbox{order})$;   //initial  $\B$ which is a permutation matrix and satisfies \eqref{e:ordernormR},
 \STATE $\bar{\L}=\B, \U=\I_n$ // $\bar{\L}, \U$ satisfy $\bar{\L}\B=\I_n$;
   \STATE $\bxi= \0$; // $\xi_i=\sum_{j=i+1}^{n}\bar{r}_{jj}^2(b_j-d_j)^2$ (see \eqref{e:levelk})
   \STATE $\s = \0$; // used to update the entries of $\b$
 \STATE $i = n$;
 \STATE  $\gamma =0$; 

    \WHILE{1}
        \STATE $\eta = \xi_i+\gamma^2$;
        \IF{$\eta<\beta^2$} 
          \IF{$k\neq1$}
           \STATE $i=i-1$;
           \STATE $\xi_i=\eta$;
           \STATE $d_i=-\frac{1}{\bar{r}_{ii}}\sum_{j=i+1}^{n}\bar{r}_{ij}b_j$; // see \eqref{e:d}
           \STATE $b_i=\lfloor d_i\rceil$;
           \STATE $\gamma = \bar{r}_{ii}(d_i-b_i)$;
           \STATE $s_i=\mbox{sgn}(d_i-b_i)$; //(see \eqref{e:sgn})
          \ELSE  
                  \IF{$\eta\neq0$}   
                   \STATE use Alg. \ref{a:updatingC}  to update $\bar{\L}$, $\U$, $\B$ and
                   $\rr$, where $\alpha=\sqrt{\eta}$, set $\beta=r_n$;
                   \STATE set $i=i+1$;
                   \ENDIF

                   \STATE $b_i = b_i+ s_i$;
                   \STATE $\gamma = \bar{r}_{ii}(d_i-b_i)$;
                   \STATE $s_i=-s_i-$sgn$(s_i)$;
          \ENDIF  
        \ELSE  
           \IF{$i=n$}
              \STATE break;
           \ELSE

             \STATE $i=i+1$;
             \IF{$i=n$}
               \STATE $b_i=b_i+1$;
              \ELSE
                 \IF{$\b_{i+1:n}=\0$}
                 \STATE $b_i=b_i+1$\
                 \ELSE
                  \STATE $b_i=b_i+s_i$;
                  \STATE $s_i=-s_i-$sgn$(s_i)$;
                 \ENDIF
             \ENDIF
             \STATE $\gamma = \bar{r}_{ii}(d_i-b_i)$;
           \ENDIF
        \ENDIF
       \ENDWHILE
 \end{algorithmic}
 \end{algorithm}

\begin{remark}
Note that the differences between Alg. \ref{a:RSMP} and the improved Schnorr-Euchner search algorithm
in \cite{WenC17b} are lines 4-10 and line 27
which are for initialization and updating suboptimal solutions $\B$, respectively.
More specifically, lines 4-10 and line 27 should be respectively changed to $\bar{\b}=\0$
(intermediate solution), and $\bar{\b}=\b, \beta=\eta$ for the improved Schnorr-Euchner search algorithm.
\end{remark}

If $\B^{\star}$ is a solution to Problem \ref{P:RSMP},
then $\A^{\star}=\Z\B^{\star}$ is a solution to the SMP on lattice $\mathcal{L}(\H)$,
where $\Z$ is defined in \eqref{e:QRZ},
thus the algorithm for $\A^{\star}$ is described in Alg. \ref{a:SMP}.

\begin{algorithm}[!ht]
\caption{A sphere decoding based algorithm for the SMP on lattice $\mathcal{L}(\H)$}
\textbf{Input:} A nonsingular upper triangular matrix $\R\in \mathbb{R}^{n\times n}$.

\textbf{Output:} A solution $\A^{\star}$ to SMP on lattice $\mathcal{L}(\H)$.

\label{a:SMP}
\begin{enumerate}
\item Perform the QR factorization to $\H$ to get an invertible matrix $\R$ (see \eqref{e:QR}).
\item Perform the LLL reduction to $\R$ to get $\bbR$ and $\Z$ (see \eqref{e:QRZ}).
\item Get $\B^{\star}$ by solving Problem \ref{P:RSMP} with Alg. \ref{a:RSMP}.
\item Set $\A^{\star}:=\Z\B^{\star}$.
\end{enumerate}
\end{algorithm}

\subsection{Optimality of the new algorithm}
\label{SS:Opt}

In this subsection, we show that the new algorithm exactly solves the SMP on lattice $\mathcal{L}(\H)$.
Since $\A^{\star}=\Z\B^{\star}$, it is equivalent to show that Algorithm \ref{a:RSMP}
exactly solves Problem \ref{P:RSMP}. Specifically, we have the following theorem which
shows the optimality of Algorithm \ref{a:RSMP}.

\begin{theorem}
\label{t:optim}
Suppose that $\B^{\star}\in \mathbb{Z}^{n\times n}$ is the full-rank matrix returned by
Algorithm \ref{a:RSMP}, then
\[
\|\bbR \b^{\star}_i\|_2=\lambda_i(\bbR),\quad \quad i=1,2,\ldots, n,
\]
where $\bbR$ is defined in \eqref{e:QRZ}.
\end{theorem}
\begin{proof}
Please see Appendix \ref{ss:optim}.
\end{proof}

\section{Complexity Analysis of the new algorithm}
\label{S:compAna}

In this section, we first theoretically show that the memory complexity and the expected time
complexity of our new SMP algorithm  are respectively $\bigO(n^2)$ space
and $\bigO(n^{3/2}(2\pi e\,m/n)^{n/2})$ flops.
Then, we show that our new SMP algorithm is $\Omega(n)$ times faster
than \cite[Alg. 2]{DinKWZ15} whose memory complexity is $\bigO(n^2)$ space,
and explain that it is also faster than \cite[Alg. 1]{FisCS16}
whose memory complexity is exponential in $n$.

\subsection{Complexity analysis of the proposed algorithm}
\label{SS:compAnaA2}

In this subsection, we analyze the space and time complexities of  Alg. \ref{a:SMP}.

We first look at its memory complexity. One can easily see that the space complexities of
both Alg. \ref{a:updatingC} and  Alg. \ref{a:RSMP} are $\bigO(n^2)$ space.
The space complexities of the QR factorization, the LLL reduction and saving $\Z$ (see \eqref{e:QRZ})
are $\bigO(n^2)$ space, thus the total memory complexity of Alg. \ref{a:SMP} is $\bigO(n^2)$ space.

In the following, we investigate the time complexity, in terms of flops, of Alg. \ref{a:SMP}.
Since the complexities of the QR factorization and computing $\A^{\star}=\Z\B^{\star}$,
and the expected complexity of the LLL reduction (when $1/4<\delta<1$) \cite{LinMH13}
are polynomial in $n$, while the complexity of Alg. \ref{a:RSMP} is exponential,
the complexity of Alg. \ref{a:SMP} is dominated by Alg. \ref{a:RSMP}.

In the sequel, we study the complexity of Alg. \ref{a:RSMP}.
From  Alg. \ref{a:RSMP}, one can see that its complexity, denoted by $C(n)$,
consists of two parts: the complexities of finding and updating integer vector
$\b$ satisfying \eqref{ine:ellp},
and updating $\B, \L, \U, \rr$ whenever a  nonzero integer vector $\b$ is obtained
(line 27 of Alg. \ref{a:RSMP}).
Let $C_1(n)$ and $C_2(n)$ respectively denote them, then
\[
C(n)= C_1(n)+C_2(n).
\]

Let $\mu_i(n)$ and $f_i$, $1\leq i\leq n$, respectively denote the number of integer vectors
$\b\in \mathbb{Z}^{n\times n}$ (see \eqref{e:levelk}) searched by the Schnorr-Euchner enumeration algorithm
and the number of flops
that the enumeration performs for searching an integer vector $\b$ in the $i$-th level.
Then by \cite{HasV05}, $f_i=\bigO(n)$ (which can be seen from Alg. \ref{a:RSMP}), and thus
\begin{align*}
C_1(n)=\sum_{i=1}^{n}\mu_i(n)f_i\leq \sum_{i=1}^{n}\mu_1(n)\bigO(n)=\bigO(n^2\mu_1(n)).
\end{align*}

Since the number of times that $\B$ needs to be updated is $\mu_1(n)$, and by the complexity
analysis of Alg. \ref{a:updatingC}, each updating costs $\bigO(n^2)$ flops, so we obtain
\[
C_2(n)=\mu_1(n)\bigO(n^2)=\bigO(n^2\mu_1(n)).
\]

By the aforementioned three equations, we have
\beq
\label{e:comp}
C(n)= \bigO(n^2\mu_1(n)).
\eeq
To compute $C(n)$, we need to know $\mu_1(n)$, but unfortunately,
exactly computing $\mu_1(n)$ is very difficult if it is not impossible.
However, from  \cite{GruW93,BanK98,AgrEVZ02,HasV05},
the expected value of $\mu_1(n)$, i.e., $E[\mu_1(n)]$, is proportional to
\[
\frac{\pi^{n/2}}{\Gamma(n/2+1)}\beta^{n},
\]
where $\beta=\|\bbR \b_n\|_2$ (see step 1 of Alg. \ref{a:RSMP}).
Note that the above strategy has also been employed in \cite{DinKWZ15} to
analyze the complexity of its algorithm.

To compare the time complexity of our proposed algorithm with that of the SMP algorithm
in \cite{DinKWZ15}, we make the same assumption as that in \cite{DinKWZ15} on $\H$,
i.e., assuming that the entries of $\H$ independently and identically follow the standard
Gaussian distribution.
Since the initial $\B$ is an $n\times n$ permutation matrix,
if we do not use the LLL reduction to reduce $\R$ in \eqref{e:QR},
then $E[\|\R\b_i\|_2]=\sqrt{m}$ for $1\leq i\leq n$.
Since the LLL reduction can generally significantly reduce the initial radius,
the expected value of the initial radius of Alg. \ref{a:RSMP} is less
than $\sqrt{n}$. Thus,
\beq
\label{e:NtE}
E[\mu_1(n)]\leq\bigO\left(\frac{(m\pi)^{n/2}}{\Gamma(n/2+1)}\right).
\eeq

By the Stirling's approximation and the fact that $\Gamma(n+1)=n!$ for any positive integers $n$,
we obtain
\[
\frac{(m\pi)^{n/2}}{\Gamma(n/2+1)}\approx \frac{1}{\sqrt{n\pi}}\left(\frac{2\pi e\,m}{n}\right)^{n/2}.
\]
Hence,
\[
E[\mu_1(n)]\lessapprox\bigO\left(\frac{1}{\sqrt{n}}\left(\frac{2\pi e\,m}{n}\right)^{n/2}\right)
\]
which combing with \eqref{e:comp} yields
\beq
\label{e:compE}
E[C(n)]\lessapprox \bigO(n^{3/2}(2\pi e\,m/n)^{n/2}).
\eeq

\subsection{Comparison of the complexity of the proposed method with that of the algorithm in \cite{DinKWZ15}}
\label{SS:compD}

In this subsection, we show that Alg. \ref{a:SMP} is $\Omega(n)$ times faster
than the SMP algorithm in  \cite{DinKWZ15}.

Note that two algorithms, which are respectively for real and complex SMP's,
were proposed in \cite{DinKWZ15}.
In this paper, we only developed an algorithm for the real SMP
since an algorithm for the complex SMP can be similarly designed and a general complex
SMP can be easily converted into an equivalent real SMP.

To better understand the real SMP algorithm in \cite{DinKWZ15}, i.e., \cite[Alg. 2]{DinKWZ15},
we briefly review it here.
It first performs the LLL reduction to $\H$, i.e., finding a unimodular matrix $\Z\in \mathbb{Z}^{n\times n}$
such that $\H\Z$ is LLL reduced. Then it performs QR factorization to $\H\Z$ to get an upper triangular
matrix $\R$ to transforms the SMP on $\mathcal{L}(\H)$ to Problem \ref{P:RSMP}.
Note that these two steps are equivalent to the first two steps of Alg. \ref{a:SMP}.
Then it solves Problem \ref{P:RSMP} to get $\B^{\star}$.
Finally it returns $\A^{\star}=\Z\B^{\star}$, where $\Z$ is defined in \eqref{e:QRZ}.
As in \cite[Alg. 1]{FenSK13}, $\B^{\star}$ is obtained column by column in $n$ iterations. To be more concrete, the solution of the SVP \eqref{e:SVPR}
forms the first column of $\B^{\star}$; for $2\leq k\leq n$,
the integer vector which minimizes $\|\bbR\b\|_2$ over all the integer vectors $\b$
that are independent with the first $k-1$ columns of $\B^{\star}$ forms the $k$-th column
of $\B^{\star}$.
These vectors are obtained by a modified Schnorr-Euchner algorithm \cite{DinKWZ15}.

By the above analysis, one can see that the memory complexity of \cite[Alg. 2]{DinKWZ15} is also
$\bigO(n^{2})$ space. So it has the same memory complexity as Alg. \ref{a:SMP}.

In the following, we compare  their time complexities.
By \cite[eqs. (15) and (18)]{DinKWZ15}, the complexity of \cite[Alg. 2]{DinKWZ15} is bounded by
\beq
\label{e:compDbd}
\bigO\left(n^{4}\frac{(m\pi)^{n/2}}{\Gamma(n/2+1)}\right).
\eeq
While, by \eqref{e:comp} and \eqref{e:NtE}, the complexity of Alg. \ref{a:SMP} is  bounded by
\beq
\label{e:compNbd}
\bigO\left(n^{2}\frac{(m\pi)^{n/2}}{\Gamma(n/2+1)}\right).
\eeq

Since \eqref{e:compDbd} may be a loose bound on the complexity of \cite[Alg. 2]{DinKWZ15},
Alg. \ref{a:SMP} may not be $\bigO(n^2)$ times faster than \cite[Alg. 2]{DinKWZ15}.
Therefore, in the following, we compare their complexities from another point of view.

By the above analysis, \cite[Alg. 2]{DinKWZ15} solves an SVP to get $\b^{\star}_1$
and $(n-1)$ variants of SVP to get $\b^{\star}_k$ for $2\leq k\leq n$.
The complexity of obtaining $\b^{\star}_k$ in \cite{DinKWZ15} is
$\bigO(nk^2)$ times of that of solving an SVP (since independence needs to be checked),
and thus the total complexity is around $\bigO(n^4)$ times of that of solving an SVP.
Since these $(n-1)$ variants of SVP may have different initial radii,
the true complexity may be lower than $\bigO(n^4)$ times of that of solving an SVP with
the largest initial radius $r_{max}$ which is defined as
\beq
\label{e:rmax}
r_{max}=\max_{1\leq i\leq n}\|\bar{\rr}_i\|_2.
\eeq
To get $\b^{\star}_n$, a variant of SVP with the initial radius $r_{max}$ needs to be solved.
Since independence needs to be checked, the complexity of obtaining $\b^{\star}_n$ is
$\bigO(n^3)$ times of that of solving an SVP with the initial radius $r_{max}$.
Therefore, the complexity of \cite[Alg. 2]{DinKWZ15} is $\Omega(n^3)$ times of that of
solving an SVP with the initial radius $r_{max}$.
In contrast, by the analysis in the above subsection, the complexity of Alg. \ref{a:SMP}
is $\bigO(n^2)$ times of that of solving an SVP with the initial radius $r_{max}$.
Hence, the new algorithm is $\Omega(n)$ times faster than
\cite[Alg. 2]{DinKWZ15}.

\subsection{Comparison of the complexity of the proposed method with that of the algorithm in \cite{FisCS16}}
\label{SS:compF}

In this subsection, we compare the complexity of Alg. \ref{a:SMP} with that of the
SMP algorithm in \cite{FisCS16}, i.e., \cite[Alg. 1]{FisCS16}.

To better understand \cite[Alg. 1]{FisCS16}, we briefly review it here.
This algorithm was designed for IF receiver design.
After obtaining an upper triangular matrix $\R$ by the Cholesky factorization,
it performs the LLL reduction to $\R$ (see \eqref{e:QRZ}) to transfer the SMP to Problem \ref{P:RSMP}.
Then it uses a matrix $\M$ to store all the integer vectors $\b$ satisfying
$\|\bbR\b\|_2\leq r_{max}$ in an nondecreasing order according to $\|\bbR\b\|_2$,
where $r_{max}$ is defined in \eqref{e:rmax}.
These $\b$'s are obtained by using the Alg.  ALLCLOSESTPOINTS in \cite{AgrEVZ02}.
Note that, as mentioned in \cite{FisCS16}, apparently linearly dependent vectors
(those multiplied by $-1$ ) are not stored in $\M$.
After this, $\M$ is transformed into a row echelon form by using the Gaussian elimination,
and then the first $n$ independent columns of $\M$ are selected to form $\B^{\star}$.
Finally it returns $\A^{\star}=\Z\B^{\star}$, where $\Z$ is defined in \eqref{e:QRZ}.

As stated in \cite{FisCS16}, the column number of $\M$ can be approximated by
$\xi$, where
\beq
\label{e:eta}
\xi=(\pi r_{max}^2)^n/(n!|\det(\R)|^2),
\eeq
where $r_{max}$ is defined in \eqref{e:rmax}.
Thus, the memory complexity of this algorithm is exponential in $n$.
Hence, it is higher than that of Alg. \ref{a:SMP} whose memory complexity is $\bigO(n^2)$ space.

In the following, we compare their time complexities in terms of flops.
Since $\M$ is obtained by the Alg.  ALLCLOSESTPOINTS in \cite{AgrEVZ02},
let $\zeta_1$ denote the number of nonzero integer vectors searched by Alg.  ALLCLOSESTPOINTS
in \cite{AgrEVZ02}, 
then by the above analysis, the complexity of \cite[Alg. 1]{FisCS16} is dominated by
using the Gaussian elimination to reduce $\M$ into a row echelon form.
Thus, the complexity is around $n^2\zeta_1$ flops (note that $\zeta_1$ is much larger than $n$).

By Alg. \ref{a:RSMP}, its initial radius is $r_{max}$ which is defined in \eqref{e:rmax}.
Different from \cite[Alg. 1]{FisCS16}, which needs to search all the integer vectors $\b$ satisfying
$\|\bbR\b\|_2< r_{max}$,
Alg. \ref{a:RSMP} searches part of them since the radius will become smaller and smaller during
the search process.
Let $\zeta_2$ denote the number of nonzero integer vectors $\b$ searched by Alg. \ref{a:RSMP},
then although we are unable to quantify the gap between $\zeta_2$ and $\zeta_1$,
by the above analysis, $\zeta_2<\zeta_1$.

By the complexity analysis of Alg. \ref{a:updatingC}, the complexity of Alg. \ref{a:SMP}
is at most around $5n^2\zeta_2$ flops. As stated in Sec. \ref{SS:Preliminary}, the true complexity
of Alg. \ref{a:updatingC} can be much less than $5n^2$ flops, thus it is expected that
Alg. \ref{a:SMP} is more efficient than \cite[Alg. 1]{FisCS16}.
Indeed, this is true, for more details, see the simulation results in Sec. \ref{S:Sim}.

\subsection{Comparison of the complexity of the proposed method with that of the Minkowski
reduction algorithm}
\label{SS:compM}

As the Minkowski reduction \cite{Min96} has been used in \cite{SakHV13} and \cite{DinKWZ15} to sub-optimally solve the SMP,
in this subsection, we compare the complexity of Alg. \ref{a:SMP} with that of the Minkowski reduction algorithm in \cite{ZhaQW12}.
Two Minkowski reduction algorithms were proposed in \cite{ZhaQW12}.
Although the second one is faster, their expected asymptotic complexities are the same.
By \cite[eq. (18) and (31)]{ZhaQW12}, the complexity is bounded by
\[
\bigO\left(n^{3}\frac{(m\pi)^{n/2}}{\Gamma(n/2+1)}\right).
\]
By \eqref{e:compNbd}, our new algorithm has smaller bound, so it is expected that
our new algorithm is faster.
Indeed, this is true, for more details, see the simulation results in Sec. \ref{S:Sim}.

By the Minkowski reduction algorithm in \cite{ZhaQW12}, one can see that its memory complexity is $\bigO(n^2)$.
Based on the above analysis, we have the following table which summarizes the optimality,
space complexities and upper bounds on the expected time complexities of ``Mink", ``DKWZ", ``FCS" and ``New Alg."
which respectively denote the Minkowski reduction algorithm in \cite[Alg. M-RED-2]{ZhaQW12}, \cite[Alg. 2]{DinKWZ15},
\cite[Alg. 1]{FisCS16} and our proposed algorithm, where $\xi$ is defined in \eqref{e:eta}.

\begin{table*}[htbp!]
\caption{The optimality, space complexities and upper bounds on the expected time complexities of ``Mink", ``DKWZ", ``FCS" and ``New Alg."}
\centering
\begin{tabular}{||c||c|c|c||}
 \hline
\backslashbox{Method}{} &  optimality& memory complexity  &  time complexity  \\ \hline
Mink \cite[Alg. M-RED-2]{ZhaQW12} & suboptimal & $\bigO(n^2)$  & $\bigO\left(n^{3}\frac{(m\pi)^{n/2}}{\Gamma(n/2+1)}\right)$  \\ \hline
DKWZ \cite[Alg. 2]{DinKWZ15} & optimal & $\bigO(n^2)$  & $\bigO\left(n^{4}\frac{(m\pi)^{n/2}}{\Gamma(n/2+1)}\right)$  \\ \hline
FCS \cite[Alg. 1]{FisCS16} & optimal & $\bigO(n^2\xi)$  & $\bigO(n^2\xi)$  \\ \hline
New Alg. & optimal & $\bigO(n^2)$  & $\bigO\left(n^{2}\frac{(m\pi)^{n/2}}{\Gamma(n/2+1)}\right)$  \\ \hline
\end{tabular}
\label{tb:1}
\end{table*}

\section{Simulation results}
\label{S:Sim}

In this section, we provide numerical results to compare the efficiency and effectiveness of our
proposed algorithm  with those of \cite[Alg. 2]{DinKWZ15}, \cite[Alg. 1]{FisCS16} and the Minkowski
reduction algorithm \cite{ZhaQW12} for solving the SMP on lattice $\mathcal{L}(\H)$ over 1000 samples.
To simplify notation, these algorithms are respectively denoted by ``New Alg.", ``DKWZ", ``FCS" and ``Mink".
We do not compare them with the SMP algorithms
in \cite{WeiC12a} and \cite{MejO13} since they are not the state-of-the-art.

For simplicity, we assume $\H$'s are $n\times n$ matrices for $n=2:2:20$.
For any fixed $n$, we first generate 1000 realizations of $\H$,
whose entries independently and identically follow the standard Gaussian distribution,
to generate 1000 SMP's on $\mathcal{L}(\H)$.
Then, we respectively use ``DKWZ", ``FCS" and ``New Alg." to solve these SMP's.
In the test, we found that it may take several hours to use the Minkowski
reduction algorithm \cite{ZhaQW12} to suboptimally solve the SMP when $n\geq 16$.
Hence, we did not use this algorithm to solve the SMP for $n=16,18, 20$.
Note that the code for this algorithm was provided by the first author of \cite{ZhaQW12},
and the same problem also exists for the HKZ reduction algorithm developed in  \cite{ZhaQW12}
(for more details, please see \cite{WenC15}).

We first compare the solution $\A$'s returned by these four algorithms.
Since the aim of solving the SMP on $n$-dimensional lattice $\mathcal{L}(\H)$ is to get an
$\A\in \mathbb{Z}^{n\times n}$ such that $\|\H\a_i\|$ is as small as possible for $1\leq i\leq n$,
we respectively use $\u^{(1)}, \u^{(2)}, \u^{(3)}, \u^{(4)}\in \mathbb{R}^{n}$ with
$u_i^{(j)}=\|\H\a_i^{(j)}\|, 1\leq i\leq n, 1\leq j\leq 4,$ to denote the lengths
of the lattice vectors $\H\a_i^{(j)}$ for the solutions
$\A^{(1)}, \A^{(2)}, \A^{(3)}, \A^{(4)}$ returned by ``Mink", ``FCS" ``New Alg." and ``DKWZ".
Figure~\ref{fig:rate} shows average $\frac{\|\u^{(1)}-\u^{(4)}\|_2}{\|\u^{(4)}\|_2}$,
$\frac{\|\u^{(2)}-\u^{(4)}\|_2}{\|\u^{(4)}\|_2}$ and $\frac{\|\u^{(3)}-\u^{(4)}\|_2}{\|\u^{(4)}\|_2}$,
which are respectively denoted by ``Mink--DKWZ", ``FCS--DKWZ", and ``New Alg.--DKWZ",
over 1000 realizations versus $n$.

From Figure \ref{fig:rate}, we can see that the average relative differences between the solutions
returned by ``FCS" and ``DKWZ", and ``New Alg." and ``DKWZ" are 0 for $n=2:2:20$, which is because they are optimal SMP algorithms.
Figure \ref{fig:rate} also shows that the average relative differences between the solutions returned by
``Mink" and ``DKWZ" tends to get larger as $n$ becomes larger.
This is because different from the latter, the former is a suboptimal SMP algorithm.

We then compare the complexities of these four algorithms.
Figure \ref{fig:flops1} displays the average numbers of flops taken by the four algorithms.
Figure \ref{fig:flops2} shows the average ratios of the numbers of flops needed by ``Mink", ``DKWZ"
and ``FCS" relative to that of ``New Alg.".
From Figures \ref{fig:flops1} and \ref{fig:flops2}, one can see that the suboptimal algorithm ``Mink"
is faster than ``DKWZ", but it is slower than ``FCS", and ``New Alg." is the most efficient
one among the four algorithms under consideration. .

\begin{figure}[!htbp]
\centering
\includegraphics[width=3.4 in]{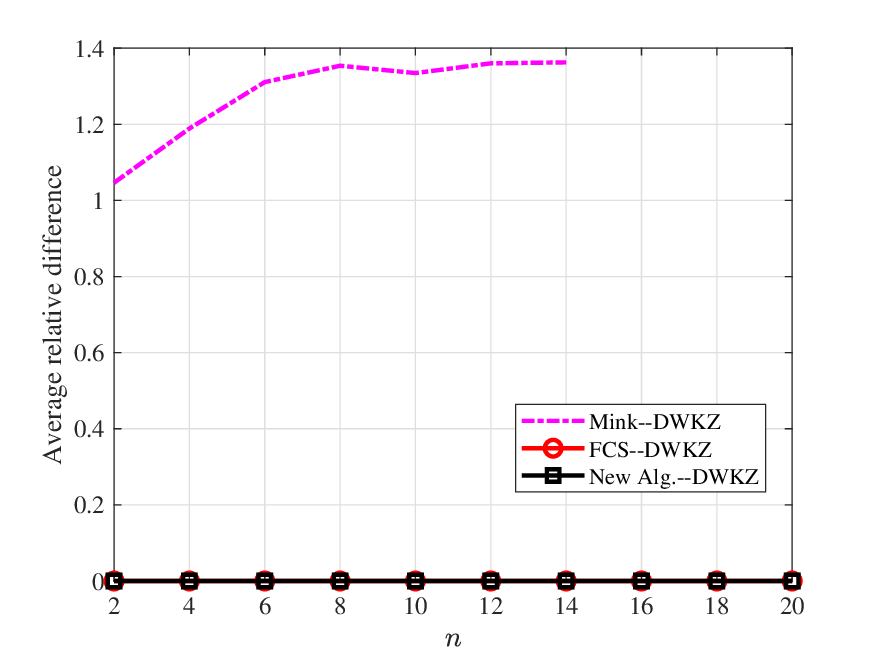}
\caption{Average relative differences over $1000$ realizations} \label{fig:rate}
\end{figure}

\begin{figure}[!htbp]
\centering
\includegraphics[width=3.4 in]{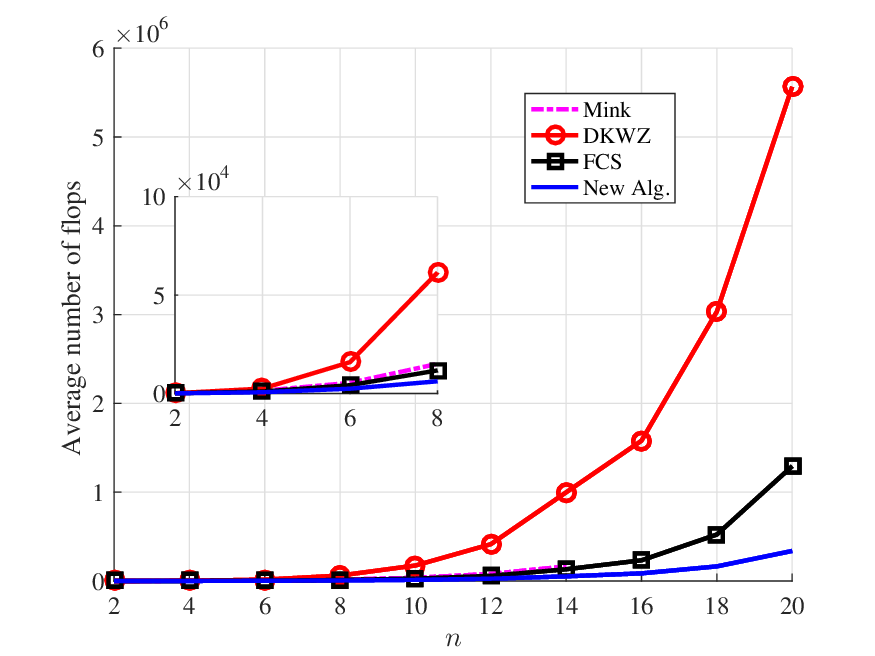}
\caption{Average number of flops over $1000$ realizations} \label{fig:flops1}
\end{figure}

\begin{figure}[!htbp]
\centering
\includegraphics[width=3.4 in]{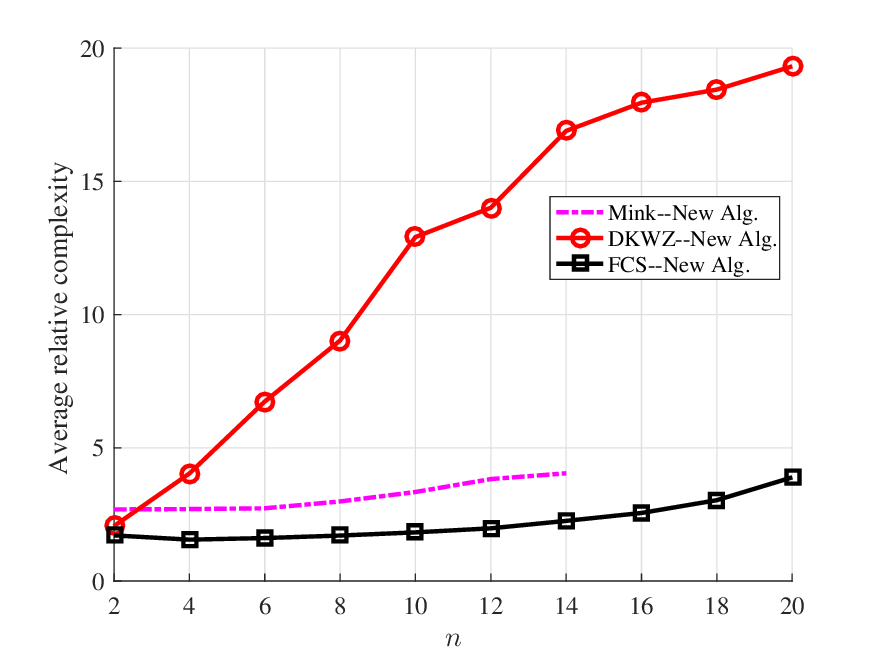}
\caption{Average ratio of flops over $1000$ realizations} \label{fig:flops2}
\end{figure}

\section{Conclusion}
\label{S:Conclusions}
In this paper, we have developed a novel efficient algorithm with an $\bigO(n^2)$
memory complexity for optimally solving an SMP on an $n$-dimensional lattice,
and theoretically showed its optimality.
Theoretical complexity analysis showed that the new algorithm
is $\Omega(n)$ times faster than the most efficient existing optimal algorithm with
polynomial memory complexity.
We have also asserted that it is faster than the most efficient existing algorithm
with exponential memory complexity.
Simulation results have also been provided to illustrate the optimality and efficiency of
the proposed algorithm.

\appendices

\section{Proof of Theorem~\ref{t:optim}}
\label{ss:optim}

To prove Theorem~\ref{t:optim}, we need to introduce the following two lemmas.
We begin with introducing the following Lemma which provides some properties of successive minima.

\begin{lemma}
\label{l:succmima1}
Suppose that there exist linearly independent vectors $\b_1,\ldots, \b_i\in \mathcal{L}(\bbR)=\{\bbR\a|\a
\in \mathbb{Z}^{n}\}$ satisfying $\|\bbR\b_j\|_2=\lambda_{j}$
for $1\leq j\leq i$ with $1\leq i \leq n$.
Then, they are the first $i$ columns of a solution of Problem \ref{P:RSMP}.
\end{lemma}

\begin{proof}
We assume $i<n$, otherwise, the lemma holds naturally.
To prove the lemma, it suffices to show that there exist
$\b_{i+1}, \ldots, \b_{n}\in \mathcal{L}(\bbR)$
such that $\b_1,\ldots, \b_{n}$ are linearly independent and
$\|\bbR\b_j\|_2=\lambda_{j}$ for $i+1\leq j\leq n$.
Let
\begin{align*}
\b_{i+1}=&\argmin_{\b\in S_i}\|\bbR\b\|_2,\\
&\vdots\\
\b_{n}=&\argmin_{\b\in S_{n}}\|\bbR\b\|_2,
\end{align*}
where
\begin{align*}
S_i=&\{\b \in \mathbb{Z}^{n} \mbox{ such that } \b, \b_1,\ldots, \b_{i}
\mbox{ are independent}\},\\
&\quad\quad\quad\quad\quad\quad\quad\quad\vdots\\
S_{n}=&\{\b \in \mathbb{Z}^{n} \mbox{ such that } \b, \b_1,\ldots, \b_{n-1}
\mbox{ are independent}\}.
\end{align*}
Then, by the proof of \cite[Theorem 8]{FenSK13}, one can see that the above
$\b_{i+1}, \ldots, \b_{n}$ satisfy the above requirements.
Note that since the solution of Problem \ref{P:RSMP} exists,
$S_i, \ldots, S_{n}$ are not empty sets. Hence, the lemma holds.
\end{proof}

By Lemma \ref{l:succmima1}, we can get the following useful lemma.
\begin{lemma}
\label{l:succmima2}
Suppose that $\lambda_{i-1}(\bbR)<\lambda_{i}(\bbR)$ for some $2\leq i \leq n$,
and  $\b$ is the $i$-th column of a solution to the Problem \ref{P:RSMP}.
Let $\bar{\b}_1,\ldots, \bar{\b}_k\in \mathcal{L}(\bbR)$ be any linearly independent vectors which
 either satisfy that
$\|\bbR\bar{\b}_j\|_2< \lambda_{i}$ or
$\|\bbR\bar{\b}_j\|_2=\lambda_{i}$ and $\bar{\b}_j$ is not the $i$-th column of any solution of Problem \ref{P:RSMP}, for $1\leq j\leq k$,
then $\bar{\b}_1,\ldots, \bar{\b}_k, \b$  are linearly independent.
\end{lemma}

\begin{proof}
Since $\b$ is the $i$-th column of a solution to Problem \ref{P:RSMP},
there exists $\b_1,\ldots, \b_{i-1}\in \mathbb{Z}^{n}$ such that
$\b_1,\ldots, \b_{i-1}, \b$ are linearly independent
and $\|\bbR\b_j\|_2=\lambda_{j}<\lambda_{i}$ for $1\leq j\leq i-1$.

For any $1\leq j\leq k$, if $\|\bbR\bar{\b}_j\|_2< \lambda_{i}$, then by the
the definition of $\lambda_{i}$ and the assumption that $\lambda_{i-1}<\lambda_{i}$,
$\bar{\b}_j$ is a linear combination of  $\b_1,\ldots, \b_{i-1}$.
Otherwise, $\bar{\b}_j, \b_1,\ldots, \b_{i-1}$ are linearly independent,
then according to $\|\bbR\b_{\ell}\|_2\leq \lambda_{i-1}<\lambda_{i}$ for $1\leq \ell \leq i-1$
and $\|\bbR\bar{\b}_j\|_2< \lambda_{i}$,
one can obtain that $\lambda_{i}\leq \max\{\|\bbR\b_{i-1}\|_2, \|\bbR\bar{\b}_{j}\|_2\}<\lambda_{i}$ which is impossible.
If $\|\bar{\b}_j\|_2=\lambda_{i}$, since it is not the $i$-th column of any solution of Problem \ref{P:RSMP}, by Lemma \ref{l:succmima1}, $\bar{\b}_j$ is a linear combination of  $\b_1,\ldots, \b_{i-1}$.

Since $\bar{\b}_1,\ldots, \bar{\b}_k$  are linearly independent,
if $\bar{\b}_1,\ldots, \bar{\b}_k, \b$  are linearly dependent,
then $\b$ is a linear combination of $\bar{\b}_1,\ldots, \bar{\b}_k$,
which combing with the above analysis imply that
$\b$ is also a linear combination of $\b_1,\ldots, \b_{i-1}$,
contradicting the fact that  $\b_1,\ldots, \b_{i-1},\b$ are linearly independent.
Thus, $\bar{\b}_1,\ldots, \bar{\b}_k, \b$  are linearly independent.
\end{proof}

In the following, we prove Theorem~\ref{t:optim} by using Lemma  \ref{l:succmima2}.

\begin{proof}
By Alg. \ref{a:RSMP}, one can see that all the columns of $\B^{\star}$,
which is a solver of Problem \ref{P:RSMP} and satisfies $b_{nj}^{\star}\geq 0$ and
$b_{ij}^{\star}\geq 0$ if $\b_{i+1:n}^{\star}=\0$ for
$1\leq i\leq n-1$, $1\leq j\leq n$, can be searched by the algorithm.
Thus, to show the theorem, it suffices to show that, during the process of
Alg. \ref{a:RSMP}, for $1\leq j\leq n$,
the first $\b_{j}^{\star}$ (note that there may exist several $\B^{\star}$'s which
are the solutions of Problem \ref{P:RSMP}, and here ``the first $\b_{j}^{\star}$" means the first vector
that obtained by Alg. \ref{a:RSMP} and is the $j$-th column of a solver of Problem \ref{P:RSMP})
will replace a column of $\B$ corresponding to the suboptimal solution when
$\b_{j}^{\star}$ is obtained by Alg. \ref{a:RSMP}
(in the following, we assume this $\b_{j}^{\star}$ is not a column
of the $n\times n$ identity matrix, otherwise we only need to show the next step),
and it will not be replaced by any vector $\b\in \mathbb{Z}^{n}$.

We first show the conclusion holds for $\b_{1}^{\star}$.
Clearly, the matrix $\B$ corresponding to the suboptimal solution when the first $\b_{1}^{\star}$
is obtained satisfies $\|\bbR \b_j\|_2>\|\bbR \b^{\star}_1\|_2$ for $1\leq j\leq n$,
thus, by Remark \ref{co:iindep} and  Alg. \ref{a:RSMP}, the first $\b_{1}^{\star}$ will replace
a column of $\B$.
Moreover, as there is not any vector $\b\in \mathbb{Z}^{n}$ satisfying
$\|\bbR\b\|_2<\|\bbR\b_{1}^{\star}\|_2$, by  Alg. \ref{a:RSMP},
$\b_{1}^{\star}$ will not be replaced by any vector $\b\in \mathbb{Z}^{n}$
which is not the $j$-th column of any solution of Problem \ref{P:RSMP}.

In the following, we show the conclusion holds for the first $\b_{j}^{\star}$ for any
$2\leq j\leq n$ with Lemma \ref{l:succmima2} by considering two cases:
$\lambda_{j-1}<\lambda_{j}$ and $\lambda_{j-1}=\lambda_{j}$.

Suppose that $\lambda_{j-1}<\lambda_{j}$. Let $\B$ be the suboptimal solution when the first $\b_{j}^{\star}$ is obtained.
Since $\B$ is full-rank and $\b_{j}^{\star}$ is the first vector that it is
the $j$-th column of a solution of Problem \ref{P:RSMP},  by Lemma \ref{l:succmima2} and Alg. \ref{a:RSMP}, $\b_{j}^{\star}$ will replace a column of $\B$.
Moreover, by Lemma \ref{l:succmima2}, all the linearly independent vectors $\b\in\mathbb{Z}^{n}$
either satisfy that $\|\bbR\b\|_2<\|\bbR\b_{j}^{\star}\|_2$ or $\|\bbR\b\|_2=\|\bbR\b_{j}^{\star}\|_2$
and $\b$ is not the $i$-th column of any solution to Problem \ref{P:RSMP} are
linearly independent with $\b_j^{\star}$, thus by Alg. \ref{a:RSMP}, $\b_{j}^{\star}$ will not be replaced by any vector $\b\in \mathbb{Z}^{n}$
which is not the $j$-th column of any solution of Problem \ref{P:RSMP}.

Suppose that $\lambda_{j-1}=\lambda_{j}$. We consider two cases: all the first $j-1$ successive minima of lattice $\mathcal{L}(\bbR)$ are equal,
and at least two of the successive minima of lattice $\mathcal{L}(\bbR)$ are different.

We first consider the first case. In this case, $\lambda_{1}=\ldots=\lambda_{i}$.
By the above analysis, Alg. \ref{a:RSMP} can find the first $\b_{1}^{\star}$,
use it to replace the first column of the suboptimal solution corresponding to $\b_{1}^{\star}$,
and will not be replaced by any other vectors.
Since $\lambda_{2}=\lambda_{1}$, there exists at least one $\b\in \mathbb{Z}^{n}$ such that
$\|\bbR\b\|_2=\|\bbR\b_{1}^{\star}\|_2$ and $\bbR\b_{1}^{\star}$ and $\bbR\b$ are linearly
independent. Thus, the first $\b_{2}^{\star}$, will replace the second column of
the suboptimal solution corresponding to this vector.
Since there is not any vector $\b\in \mathbb{Z}^{n}$ satisfying
$\|\bbR\b\|_2<\lambda_{2}=\lambda_{1}$, by Alg. \ref{a:RSMP},
$\b_{2}^{\star}$ will not be replaced by any vector $\b\in \mathbb{Z}^{n}$.
Similarly, one can show that the first $\b_{j}^{\star}$ will replace a column of $\B$
corresponding to the suboptimal solution when it is obtained.
Moreover, since there is not any vector $\b\in \mathbb{Z}^{n}$ satisfying
$\|\bbR\b\|_2<\lambda_{1}=\lambda_{i}$, by Alg. \ref{a:RSMP},
$\b_{j}^{\star}$ will not be replaced by any other vector $\b\in \mathbb{Z}^{n}$.

We then consider the second case, and let (if $j\geq 3$)
$
k=\argmin_{1\leq i \leq j-2}\lambda_{i}<\lambda_{i+1}.
$
By the above proof, one can see that Alg. \ref{a:RSMP} can find
$\b_{1}^{\star},\ldots, \b_{k}^{\star}$,
use them to replace the first $k$ columns of the suboptimal solution corresponding to them iteratively, and will not replace them with any other vectors.
Since $\lambda_{k}<\lambda_{k+1}$, the conclusion holds clearly for $\b_{k+1}^{\star}$.
Similarly and iteratively, one can show that $\b_{j}^{\star}$ will replace a column of $\B$
corresponding to the suboptimal solution when it is obtained,
and it will not be replaced by any vector $\b\in \mathbb{Z}^{n}$ which is not the $j$-th column of any solution of Problem \ref{P:RSMP}.
\end{proof}

\bibliographystyle{IEEEtran}
\bibliography{ref2}

\end{document}